\documentclass[journal]{IEEEtran}
\usepackage{array}
\IEEEoverridecommandlockouts
\usepackage{cite}
\usepackage{amsmath,amssymb,amsfonts}
\usepackage{algorithm}
\usepackage{algpseudocode}
\usepackage{graphicx}
\usepackage{amsmath}
\usepackage{amsthm}
\usepackage{textcomp}
\usepackage{xcolor}
\usepackage{amsmath}
\usepackage{float} 
\usepackage{comment}
\usepackage{subfigure} 
\usepackage{color}
\usepackage{bm}
\usepackage{setspace}
\usepackage{esvect}

\newtheorem{lemma}{\emph{\underline{Lemma}}}
\newtheorem{corollary}{\emph{\underline{Corollary}}}

\newtheorem{proposition}{\emph{\underline{Proposition}}}

\def\phi{\varphi}

\def\({\left(}
\def\){\right)}

\def\b0{{\mathbf{0}}}

\newcommand{\diag}{\mathrm{diag}}

\usepackage[top=0.73in, bottom=0.68in, left=0.58in, right=0.58in]{geometry}

\def\BibTeX{{\mathrm B\kern-.05em{\sc i\kern-.025em b}\kern-.08em
    T\kern-.1667em\lower.7ex\hbox{E}\kern-.125emX}}

\begin{document}
\allowdisplaybreaks[4]
\renewcommand{\baselinestretch}{0.9}
\title{
\begin{flushleft}
\author{Zhenyu~Kang,~\IEEEmembership{Student~Member,~IEEE}, Changsheng~You,~\IEEEmembership{Member,~IEEE},\\
	 and Rui Zhang,~\IEEEmembership{Fellow,~IEEE} \thanks{This work is supported by Ministry of Education, Singapore under Award T2EP50120-0024 and by Advanced Research and Technology Innovation Centre (ARTIC) of National University of Singapore under Research Grant R-261-518-005-720. Z. Kang and R. Zhang are  with the Department of Electrical and Computer Engineering, National University of Singapore, Singapore (Email: zhenyu\_kang@u.nus.edu, elezhang@nus.edu.sg). C. You is with the Department of Electrical and Electronic Engineering, Southern University of Science and Technology (SUSTech), Shenzhen 518055 , China. He was with the Department of Electrical and Computer Engineering, National University of Singapore, Singapore 117583 (e-mail: eleyouc@nus.edu.sg).
}\vspace{-20pt}}
\end{flushleft}
\huge 
IRS-Aided Wireless Relaying: \\{\color{black}Deployment Strategy and Capacity Scaling}} 
\maketitle

\begin{abstract}
In this letter, we consider an intelligent reflecting surface (IRS)-aided wireless relaying system, where a decode-and-forward relay (R) is employed to forward data from a source (S) to a destination (D), aided by $M$ passive reflecting elements.
We consider two practical IRS deployment strategies, namely, \textit{single-IRS deployment} where all reflecting elements are mounted on one single IRS that is deployed near S, R, or D, and \textit{multi-IRS deployment} where the reflecting elements are allocated over three separate IRSs which are deployed near S, R, and D, respectively. 
Under the line-of-sight (LoS) channel model, we characterize the capacity scaling orders with respect to an increasing $M$ for the IRS-aided relay system with different IRS deployment strategies.
For single-IRS deployment, we show that deploying the IRS near R achieves the highest capacity as compared to that near S or D.
{\color{black}While for multi-IRS deployment, we propose a practical cooperative IRS passive beamforming design which is analytically shown to achieve a larger capacity scaling order than the single-IRS deployment (i.e., near R or S/D) when $M$ is sufficiently large.}
Numerical examples are provided, which validate our theoretical results.
\end{abstract}
\vspace{-8pt}
\section{Introduction}
\begin{spacing}{0.84}
Intelligent reflecting surface (IRS) has emerged as a promising technology to enhance the spectral/energy efficiency of future wireless communication systems  \cite{9326394,8796365}. Specifically, with a large array of passive reflecting elements, IRS is able to control the radio propagation environment by dynamically tuning the amplitudes and/or phases of incident signals.
Thus, IRS has recently received significant research interests and been investigated in various wireless systems/applications\cite{9326394,8796365}. 

In the existing literature, IRS is mainly employed as a passive relay near the transmitter/receiver to assist their communication \cite{8811733,9110912}.
Different from the conventional half-duplex amplify-and-forward (AF) and decode-and-forward (DF) active relays, IRS operates in full-duplex with passive signal reflection only, thus incurring no amplification/processing noise, and achieving high spectral efficiency at low hardware/energy cost \cite{8811733,9110912,8741198}. 
Recently, to further improve the wireless relaying performance, new hybrid active/passive relay systems have been proposed in \cite{obeed2021joint,9108262,ying2020relay}, where IRS is employed to enhance the channel gain of the transmitter-(active) relay and/or relay-receiver links.
However, these works assumed a single IRS (or multiple IRSs in close proximity) deployed at fixed location, but did not study the effect of IRS deployment (i.e., IRS elements partition/placement) on the system performance.
For example, it remains unknown whether the conventional strategy by deploying IRS near the transmitter/receiver is still optimal for the IRS-aided relay system, and whether it is beneficial to divide one single IRS into multiple smaller-size IRSs and place them near the transmitter, receiver, and/or relay. 

To address the above questions, we study in this letter an IRS-aided wireless relaying system, where a DF relay (R) is employed to forward data from a source (S) to a destination (D), aided by a total number of $M$ reflecting elements.
{\color{black}We consider two IRS deployment strategies of practical interests, namely, \textit{single-IRS deployment} with all the reflecting elements mounted on one single IRS near S, R or D, and \textit{multi-IRS deployment} with the reflecting elements properly allocated over three separate IRSs which are placed near S, R, and D, respectively.}
To characterize the fundamental capacity scaling orders with respect to (w.r.t.) an asymptotically large $M$ for the considered system under different IRS deployment strategies, we assume the line-of-sight (LoS) channel model for all the available links. 
{\color{black}For single-IRS deployment, we show that deploying the IRS near R achieves the highest capacity as compared to that near S or D.
While for multi-IRS deployment, we propose a practical cooperative IRS passive beamforming design and show that it can achieve an even higher capacity than the single-IRS deployment (i.e., near R or S/D) when $M$ is sufficiently large, thanks to the higher asymptotic passive beamforming gain of the double-reflection link between IRSs, which is unavailable under the single-IRS deployment. The numerical examples show the superior performance with the help of IRS and validate our theoretical results.}

\emph{Notations}: 
Superscripts $(\cdot)^T$ and $(\cdot)^H$ stand for the transpose and Hermitian transpose operations, respectively. $\mathbb{C}^{a \times b}$ denotes the space of $a \times b$ complex-valued matrices. The operation $\arg(\cdot)$ returns the angle of the complex value, $\diag{(\boldsymbol{x})}$ returns a diagonal matrix with the elements in $\boldsymbol{x}$ on its main diagonal. $\jmath$ represents the imaginary unit, $\otimes$ denotes the Kronecker product, $[\cdot]_{m,n}$ denotes the $(m,n)$-th entry of the matrix, and $[\cdot]_{m}$ denotes the $m$-th entry of the vector.

\vspace{-10pt}
\section{System Model}
{\color{black}We consider an IRS-aided wireless relaying system as illustrated in Fig.~\ref{sysmod}, where S (e.g., access point) transmits data to D (e.g., a user within a given area) that is placed $2L$ meters (m) away from S\footnote{{\color{black}Note that the considered system setup with fixed node locations can be extended to the case where the user moves within the target area, and is also applicable when the user moves out of this target area and enters into another area which is covered by a different IRS.}}.} {\color{black}We assume that the direct link between S and D is negligible due to large $L$ and/or dense obstacles in the environment.
A DF relay is thus deployed between them to assist their communication.} Specifically, R receives message from S in one time slot, then decodes and forwards it to D in the next time slot, where we assume that all time slots are of equal duration, and S and R have the same transmit power of $P$ for the ease of exposition. 
All S, R, and D are equipped with a single antenna.
Without loss of generality, we consider a three-dimensional (3D) Cartesian coordinate system, where the locations of S and D are denoted by $\boldsymbol{u}_{\rm S}=(0,0,0)$ and $\boldsymbol{u}_{\rm D}=(2L,0,0)$, respectively. 
To balance the path-loss of the S-R and R-D links, R is deployed in the middle of S and D at the location of $\boldsymbol{u}_{\rm R}=(L,0,0)$\footnote{{\color{black}For the general case with different transmit powers of S and R, R should be properly placed in between S and D to balance the achievable rates of the S-R and R-D links. Nevertheless, it can be shown that our obtained results on the capacity scaling order still apply to this general case.}}. 
{\color{black}Moreover, to improve the S-R and/or R-D links, one or more IRSs with a total number of $M$ reflecting elements can be deployed.} 
\begin{figure}[t]
\centerline{\includegraphics[width=3.3in]{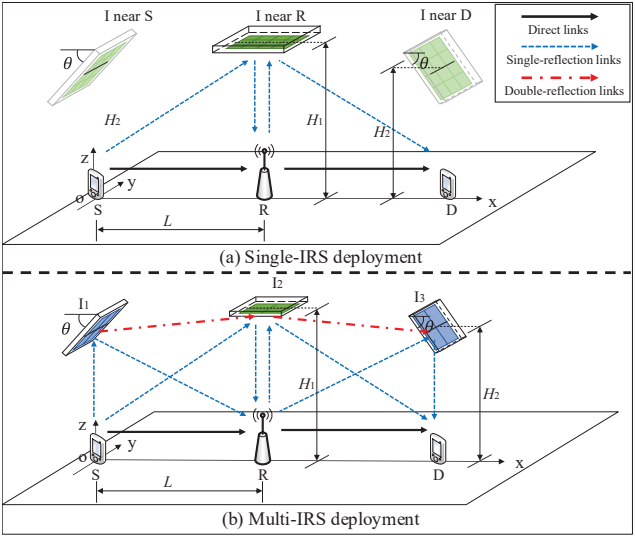}}
\vspace{-7pt}
\caption{IRS-aided wireless relaying systems.}\label{sysmod}
\vspace{-10pt}
\end{figure}

We consider two strategies for deploying the $M$ reflecting elements, namely, 
\textit{single-IRS deployment} with all the reflecting elements mounted on one single IRS that is placed near S, R, or D (see Fig. \ref{sysmod}(a)), 
and \textit{multi-IRS deployment} that distributes the $M$ reflecting elements over three separate IRSs placed near S, R, and D, respectively (see Fig.~\ref{sysmod}(b)). 
Note that to ensure there exists a single-reflection link in the case of single-IRS deployment as well as single-/double-reflection links in the case of multi-IRS deployment, we consider the following practical IRS placement near S/R/D: the IRS near R is placed parallel to the (x, y)-plane with its central point at $\boldsymbol{u}_{\rm I_{\rm R}}=(L,0,H_1)$, while the IRS near S (D) is placed towards S/R/D with a downtilt angle $\theta$ and the central point at $\boldsymbol{u}_{\rm I_{\rm S}}=(0,0,H_2)$ ($\boldsymbol{u}_{\rm I_{\rm D}}=(2L,0,H_2)$); and $H_2<H_1$ (see Fig.~\ref{sysmod}).
\vspace{-7pt}

\section{Single-IRS Deployment}
In this section, we consider the single-IRS deployment. We first present its channel model and then characterize its capacity scaling order w.r.t. an asymptotically large $M$. 

\vspace{-7pt}
\subsection{Single-IRS Deployment near R}
We first consider the case with a single IRS deployed near R.
For simplicity, we assume the LoS channel model for all the available links\footnote{\color{black}Under the LoS channel model, the CSI can be easily obtained based on the locations and orientations of S/R/D/I. Moreover, for the more general Rician fading channel model, the cascaded CSI can be practically obtained by using existing IRS channel estimation techniques under the single-IRS deployment (e.g., \cite{9133142}) or double-IRS deployment (e.g., \cite{9373363}).}.
Let $\boldsymbol{w}(\varsigma, N) \triangleq\left[1, e^{-\jmath \pi \varsigma}, \cdots, e^{-\jmath \pi(N-1) \varsigma}\right]^{T}$ denote the one-dimensional (1D) steering vector function, where $\varsigma$ is the phase difference between two adjacent reflecting elements, and $N$ denotes the size of a uniform linear array. Then the channel from node $i$ to $j$, with $i\in\{\mathrm{\text{S,I,R}}\},j\in\{\mathrm{\text{I,R,D}}\},i\neq j$, and $\{i,j\}\neq\{\text{S,D}\}$, can be modeled
in the following general form{\color{black}\footnote{\color{black}{It is worth noting that the transceivers in our considered system (with a carrier frequency of 6 GHz) are practically located in the far-field of each IRS element due to its small size ($d_{\rm{I}}=4/\lambda$ with $\lambda=0.05$ m) as compared to the distance with the nearby transceivers (e.g., $H_1=5$ m, and $H_2=4$ m); as a result, our assumed ``far-field" channel model has been verified to be very close to the actual ``far-field" model based on each individual element of the IRS.}}}
\vspace{-3pt}
\begin{align}\label{g_ij}
    \boldsymbol{G}_{{i, j}}={g}_{{i, j}}\boldsymbol{a}_{\rm r}\left(\theta^{\rm r}_{i,j}, \vartheta^{\rm r}_{i,j},m_{j}\right)\boldsymbol{a}_{\rm t}^H\left(\theta^{\rm t}_{{i, j}}, \vartheta^{\rm t}_{{i,j}},m_{i}\right),
\end{align}
where ${g}_{{i, j}}={\beta_0^{\frac{1}{2}}e^{-\jmath\frac{2\pi}{\lambda}D_{i,j}}}/{D_{i,j}^{\frac{\alpha}{2}}}$ is the complex channel gain of the $i\to j$ link with $\beta_0$ denoting its channel power gain at the reference distance of $1$ m, $\lambda$ denoting the signal wavelength, $D_{i,j}$ denoting the distance from node $i$ to $j$, and $\alpha$ denoting the path-loss exponent.
Moreover, $\theta^{\rm r}_{i,j} (\text{or } \vartheta^{\rm r}_{i,j})$ $\in$ $[0,\pi]$ and $\theta^{\rm t}_{i,j} (\text{or }\vartheta^{\rm t}_{i,j})$ $\in$ $[0,\pi]$ denote respectively the azimuth (or elevation) angle-of-arrival (AoA) at node $j$ and angle-of-departure (AoD) at node $i$ w.r.t. the IRS plane, 
$m_i$ ($m_j$) denotes the number of elements/antennas at node $i$ ($j$), 
$\boldsymbol{a}_{\rm r}\left(\theta^{\rm r}_{i,j}, \vartheta^{\rm r}_{i,j},m_{j}\right)\triangleq \boldsymbol{w}(\frac{2d_{\mathrm I}}{\lambda}\cos{\theta^{\rm r}_{{i, j}}}\sin{\vartheta^{\rm r}_{{i, j}}},M^{\rm r}_{j,\mathrm{h}})\otimes\boldsymbol{w}(\frac{2d_{\mathrm I}}{\lambda}\cos{\vartheta^{\rm r}_{{i, j}}},{M^{\rm r}_{j,\mathrm{v}}})$ denotes the receive array response from node $i$ to $j$ where $M^{\rm r}_{j,\mathrm{h}}$ and $M^{\rm r}_{j,\mathrm{v}}$ denote the number of horizontal and vertical elements of node $j$ with $M^{\rm r}_{j,\mathrm{h}}\times M^{\rm r}_{j,\mathrm{v}}=m_{j}$, 
and $\boldsymbol{a}_{\rm t}\left(\theta^{\rm t}_{{i, j}}, \vartheta^{\rm t}_{{i,j}},m_{i}\right)$ denotes the transmit array response from node $i$ to $j$ which can be defined similar to $\boldsymbol{a}_{\rm r}\left(\theta^{\rm r}_{i,j}, \vartheta^{\rm r}_{i,j},m_{j}\right)$.
Note that when $\{i,j\}$ $\subset$ $\{\text{S,R,D}\}/\{\text{S,D}\}$, we have $\boldsymbol{a}_{\rm t}$ $(\boldsymbol{a}_{\rm r})$ $=$ $1$.
Let $\boldsymbol{\Phi}_{\mathrm{I},i,j}$ $\in$ $\boldsymbol{\mathcal{P}}_M$ $\triangleq$ $\{\boldsymbol{\Phi}\big|   |[\boldsymbol{\Phi}]_{m,m}|$ $=$ $1,[\boldsymbol{\Phi}]_{m,n}$ $=$ $0,m,n$ $\in$ $\{1,\cdots,M\},m$ $\neq$ $n\}$ denote the reflection matrix of I when node $i$ transmits data to $j$.
As such, the effective channel from S to R with I deployed near R, denoted by $h^{\rm (R)}_{\mathrm{S,R}}$,\footnote{The superscripts (R), (S), (D) and (M) represent the single-IRS deployment near R, S and D, as well as the multi-IRS deployment, respectively.} is given by
\vspace{-3pt}
\begin{align}\label{hsr_R}
    h^{\rm (R)}_{\mathrm{S,R}}=g_{\mathrm{S,R}}+\boldsymbol{g}^H_{\mathrm{I, R}}\boldsymbol{\Phi}_{\mathrm{I},\mathrm{S},\mathrm{R}}\boldsymbol{g}_{\mathrm{S, I}},
\end{align}
where $g_{\mathrm{S,R}}\in\mathbb{C}$, $\boldsymbol{g}^H_{\mathrm{I, R}}\in\mathbb{C}^{1\times M}$ and $\boldsymbol{g}_{\mathrm{S, I}}\in\mathbb{C}^{M\times 1}$ denote the channels from S to R, I to R, and S to I, respectively, which are modeled based on \eqref{g_ij}.
The corresponding maximum achievable rate from S to R in bits per second per Hertz (bps/Hz), denoted by $r^{\rm (R)}_{\mathrm{S,R}}$, is thus given by
\vspace{-3pt}
\begin{equation}\label{Rsr_R}
    r^{\rm (R)}_{\mathrm{S,R}} = \log_2\left(1+\frac{P|h^{\rm (R)}_{\mathrm{S,R}}|^2}{\sigma^2}\right),
\end{equation}
where $\sigma^2$ denotes the noise power at R. Similarly, the effective channel from R to D, and the corresponding maximum achievable rate (denoted by $r^{\rm (R)}_{\mathrm{R,D}}$) can also be obtained. Based on the above, the system capacity with a single IRS deployed near R, denoted by $C^{\rm (R)}$, is given by
\vspace{-3pt}
\begin{align}\label{C_R}
    C^{\rm (R)}&\!=\!\frac{1}{2}\min{\!\left(\!\max_{\boldsymbol{\Phi}_{\mathrm{I},\mathrm{S},\mathrm{R}}} r^{\rm (R)}_{\mathrm{S,R}},\max_{\boldsymbol{\Phi}_{\mathrm{I},\mathrm{R},\mathrm{D}}}r^{\rm (R)}_{\mathrm{R,D}}\!\right)}\!\overset{(a)}{=}\!\max_{\boldsymbol{\Phi}_{\mathrm{I},\mathrm{S},\mathrm{R}}} \frac{1}{2}r^{\rm (R)}_{\mathrm{S,R}},
\end{align}
where $\boldsymbol{\Phi}_{\mathrm{I},\mathrm{S},\mathrm{R}}$, $\boldsymbol{\Phi}_{\mathrm{I},\mathrm{R},\mathrm{D}}$ $\in$ $\boldsymbol{\mathcal{P}}_M$, $(a)$ is because the links from S to R and R to D are symmetric and thus $r^{\rm (R)}_{\mathrm{S,R}}$ $=$ $r^{\rm (R)}_{\mathrm{R,D}}$ by properly designing the IRS passive beamforming (shown next),
and the factor $\frac{1}{2}$ is due to the fact that data is transmitted from S to D over two equal time slots. 
To maximize $r^{\rm (R)}_{\mathrm{S,R}}$ in \eqref{Rsr_R}, it can be shown that the main diagonal of the optimal IRS passive beamforming is \cite{8811733}
\vspace{-3pt}
\begin{equation}\label{phisr_R}
\begin{aligned}
    \!\!\!&[\boldsymbol{\Phi}_{\mathrm{I},\mathrm{S},\mathrm{R}}]_{m,m}\!\!=\!\!e^{\jmath\left({\arg\left({g_{\mathrm{S,R}}}\right)}-\arg\left({[\boldsymbol{g}^H_{\mathrm{I,R}}]_m}\right)-\arg\left({[\boldsymbol{g}_{\mathrm{S, I}}]_m}\right)\right)},
\end{aligned}
\end{equation} 
where $ m\!=\!1,\!\cdots\!,M$.
In this case, the channel power gain of the S-I-R reflection link is maximized, and the S-I-R and direct S-R channels are phase-aligned.
By substituting \eqref{hsr_R} and \eqref{phisr_R} into \eqref{Rsr_R} and \eqref{C_R}, $C^{\rm (R)}$ is obtained as

\begin{align}\label{C_R2}
    C^{\rm (R)}=\frac{1}{2}\log_2\left(1+\frac{P}{\sigma^2}\left(\frac{\beta_0^{\frac{1}{2}}}{L^{\frac{\alpha}{2}}}+\frac{M\beta_0}{H_1^{\frac{\alpha}{2}}(H_1^2+L^2)^{\frac{\alpha}{4}}}\right)^2\right).
\end{align}

Based on \eqref{C_R2}, the capacity scaling order w.r.t. an asymptotically large $M$ is characterized in the following proposition.
\vspace{-5pt}
\begin{proposition}\label{pro1}
\allowdisplaybreaks[4]
\!\emph{For single-IRS deployment near R, the system capacity increases with $M$ as $M\!\to\!\infty$ according to}
\begin{align}
    \lim_{M\to\infty}\frac{C^{\rm (R)}}{\log_2M}=1.
\end{align}
\end{proposition}
\begin{proof}
With $M\to\infty$ and the capacity $C^{\rm (R)}$ given in \eqref{C_R2}, the capacity scaling order is obtained as
\begin{equation}
\begin{aligned}
    &\lim_{M\to\infty}\!\frac{C^{\rm (R)}}{\log_2M}\\
    &\!=\!\!\lim_{M\to\infty}\!\frac{1}{2}\!\log_2\!\!\left(\!1\!+\!\frac{P}{\sigma^2}\!\!\left(\!\frac{\beta_0^{\frac{1}{2}}}{L^{\frac{\alpha}{2}}}\!+\!\frac{M\beta_0}{H_1^{\frac{\alpha}{2}}(H_1^2\!+\!L^2)^{\frac{\alpha}{4}}}\!\right)^{\!\!2}\right)/\log_2M\\
    &\!=\!\lim_{M\to\infty}\frac{\log_2M}{\log_2M}\!+\!\frac{\log_2\left(\sqrt{P}\beta_0\right)}{\log_2M}\!-\!\frac{\log_2\left(\sigma{H_1^{\frac{\alpha}{2}}(H_1^2\!+\!L^2)^{\frac{\alpha}{4}}}\right)}{\log_2M}\\
    &=1,
\end{aligned}
\end{equation}
which thus completes the proof.
\end{proof}

\vspace{-5pt}
\subsection{Single-IRS Deployment near S or D}
Next, we consider the case of single-IRS deployment near S, while by using the system symmetry, the obtained results can be readily applied to the case where I is deployed near D.
For data transmission from S to R, by assuming the LoS channel model given in \eqref{g_ij}, the effective channel can be modeled in a similar form as \eqref{hsr_R} and its achievable rate, denoted by $r^{\rm (S)}_{\mathrm{S,R}}$, can be obtained in a form similar as \eqref{Rsr_R}. 
As such, the corresponding effective channel from R to D is given by $h^{\rm(S)}_{\mathrm{R,D}}=g_{\mathrm{R,D}}$,
leading to the following achievable rate
\begin{equation}\label{Rrd_S}
    r^{\rm (S)}_{\mathrm{R,D}}=\log_2\left(1+\frac{P|g_{\mathrm{R,D}}|^2}{\sigma^2}\right)=\log_2\left(1+\frac{P\beta_0}{\sigma^2L^{\alpha}}\right).
\end{equation}
Due to the lack of the R-I-D reflection link, it can be easily shown that $r^{\rm (S)}_{\mathrm{S,R}}>r^{\rm (S)}_{\mathrm{R,D}}$ by properly designing the passive beamforming and thus the corresponding capacity, denoted by $C^{\mathrm{(S)}}$, is given by
\begin{align}\label{C_S}
    \!\!\!\!\!C^{\mathrm{(S)}}\!=\!\frac{1}{2}\!\min{\!\left(\!\max_{\boldsymbol{\Phi}_{\mathrm{I},\!\mathrm{S},\mathrm{R}}} \!r^{\rm (S)}_{\mathrm{S,R}},\max_{\boldsymbol{\Phi}_{\mathrm{I},\mathrm{R},\mathrm{D}}}\!r^{\rm (S)}_{\mathrm{R,D}}\!\right)}\!=\!\frac{1}{2}r^{\rm (S)}_{\mathrm{R,D}},
\end{align}
where $\boldsymbol{\Phi}_{\mathrm{I},\mathrm{S},\mathrm{R}}$, $\boldsymbol{\Phi}_{\mathrm{I},\mathrm{R},\mathrm{D}}\in\boldsymbol{\mathcal{P}}_M$.
Based on the above, we can easily obtain the capacity scaling order for the case with a single IRS deployed near S.
\begin{proposition}\label{pro2}
\emph{For single-IRS deployment near S, its capacity does not increase with $M$, i.e., as $M\to\infty$,}
\begin{equation}
    \lim_{M\to\infty}\frac{C^{\mathrm{(S)}}}{\log_2M}=0.
\end{equation}
\end{proposition}
Comparing Propositions \ref{pro1} and \ref{pro2} yields the following result.
\begin{corollary}\label{pro3}
\emph{For single-IRS deployment, we have}
\begin{align}\nonumber
    \lim_{M\to\infty}\frac{C^{\rm (R)}}{\log_2M}=1>\lim_{M\to\infty}\frac{C^{\mathrm{(S)}}}{\log_2M}=\lim_{M\to\infty}\frac{C^{\mathrm{(D)}}}{\log_2M}=0,
\end{align}
\emph{where $C^{\mathrm{(D)}}$ is the capacity for single-IRS deployment near D, which can be shown equal to $C^{\mathrm{(S)}}$.}
\end{corollary}
This proposition is intuitively expected as the IRS deployed near S or D can be used to help data transmission for one of the S-R and R-D links only, while deploying the IRS near R can help both the S-R and R-D links.

\section{Multi-IRS Deployment}
\vspace{-3pt}
In this section, we characterize the capacity scaling order for the case of multi-IRS deployment where the $M$ reflecting elements are distributed over three separate IRSs near S, R, and D, denoted by I$_1$, I$_2$, and I$_3$, respectively (see Fig.~\ref{sysmod}(b)). 
In this case, the data transmission from S to D can be assisted by two single-reflection links (S-I$_1$-R and S-I$_2$-R links) and one double-reflection link (S-I$_1$-I$_2$-R link). 
Let $M_k$ denote the number of reflecting elements at I$_k$. Due to the system symmetry, we set $M_1=M_3=\rho M$ and $M_2=(1-2\rho)M$ with $0<\rho<0.5$.\footnote{$\rho=0$ corresponds to the single-IRS deployment near R, while $\rho=0.5$ corresponds to the case where all reflecting elements are equally deployed near S and D, while the resulting two IRSs operate independently with no cooperation.}
Moreover, we assume the LoS channel model for all the available links.
For data transmission from S to R, the effective channel under the multi-IRS deployment, denoted by $h^{\rm (M)}_{\mathrm{S,R}}$, is given by
\begin{equation}\label{hsrc}
\begin{aligned}
    &h^{\rm (M)}_{\mathrm{S,R}}=g_{\mathrm{S,R}}+\underbrace{\boldsymbol{g}^H_{\mathrm{I_{2}, R}}\boldsymbol{\Phi}_{\mathrm{I_2,S,R}}\boldsymbol{G}_{\mathrm{I_{1}, I_{2}}}\boldsymbol{\Phi}_{\mathrm{I_1,S,R}}\boldsymbol{g}_{\mathrm{S, I_{1}}}}_{h_{\mathrm{dr}}}\\
    &\qquad\quad+\underbrace{\boldsymbol{g}^H_{\mathrm{I_{1}, R}}\boldsymbol{\Phi}_{\mathrm{I_1,S,R}}\boldsymbol{g}_{\mathrm{S, I_{1}}}}_{h_{\mathrm{sr,I}_1}}+\underbrace{\boldsymbol{g}^H_{\mathrm{I_{2}, R}}\boldsymbol{\Phi}_{\mathrm{I_2,S,R}}\boldsymbol{g}_{\mathrm{S, I_{2}}}}_{h_{\mathrm{sr,I}_2}},
\end{aligned}
\end{equation}
where the channels $\boldsymbol{G}_{\mathrm{I_{1}, I_{2}}}\in\mathbb{C}^{M_2\times M_1}$, $\boldsymbol{g}^H_{\mathrm{I}_{k}, \mathrm{R}}$ $\in$ $\mathbb{C}^{1\times M_k}$,  $\boldsymbol{g}_{\mathrm{S, I}_{k}}$ $\in$ \!$\mathbb{C}^{M_k\times 1}$ are modeled based on \eqref{g_ij}, $\boldsymbol{\Phi}_{\mathrm{I}_k,\mathrm{S,R}}$ denotes the diagonal reflection matrix of I$_k$ when S transmits data to R, $h_{\mathrm{dr}}$ denotes the double-reflection channel, and $h_{\mathrm{sr,I}_k}$ denotes the single-reflection channel over I$_k$. 
Thus, the corresponding maximum achievable rate from S to R, denoted by $r^{\rm (M)}_{\mathrm{S,R}}$, is given by
\vspace{-3pt}
\begin{equation}\label{R_sr_C}
    r^{\rm (M)}_{\mathrm{S,R}} = \log_2\left(1+\frac{P|h^{\rm (M)}_{\mathrm{S,R}}|^2}{\sigma^2}\right).
\end{equation}
Similarly, the maximum achievable rate from R to D, denoted by $r^{\rm (M)}_{\mathrm{R,D}}$, can be obtained in a form similar as \eqref{R_sr_C}.
Based on the above, the capacity of the relay system under the multi-IRS deployment, denoted by $C^{\mathrm{(M)}}$, is given by
\begin{align}\label{C_c}
    \!\!\!\!C^{\mathrm{(M)}}\!=\!\frac{1}{2}\min{\left(\!\max_{\substack{\boldsymbol{\Phi}_{\mathrm I_1,{\mathrm{S,R}}}\\ \boldsymbol{\Phi}_{\mathrm I_2,{\mathrm{S,R}}}}}\!\!
    r^{\rm (M)}_{\mathrm{S,R}},\max_{\substack{\boldsymbol{\Phi}_{\mathrm I_2,{\mathrm{R,D}}}\\ \boldsymbol{\Phi}_{\mathrm I_3,{\mathrm{R,D}}}}}\!\!r^{\rm (M)}_{\mathrm{R,D}}\!\right)}\!\overset{(b)}{=}\!\!\max_{\substack{\boldsymbol{\Phi}_{\mathrm I_1,{\mathrm{S,R}}}\\ \boldsymbol{\Phi}_{\mathrm I_2,{\mathrm{S,R}}}}}\frac{1}{2}
    r^{\rm (M)}_{\mathrm{S,R}},
\end{align}
where $\boldsymbol{\Phi}_{\mathrm{I}_k,\mathrm{R,D}}$ represents the diagonal reflection matrix of I$_k$ when R transmits data to D, $\boldsymbol{\Phi}_{\mathrm{I}_k,\mathrm{S},\mathrm{R}}$, $\boldsymbol{\Phi}_{\mathrm{I}_k,\mathrm{R},\mathrm{D}}$ $\in$ $\boldsymbol{\mathcal{P}}_{M_k}$ with $k$ $=$ $1,2,3$, 
and $(b)$ is because the S$\to$R and R$\to$D links are symmetric and thus $r^{\rm (M)}_{\mathrm{S,R}}$ $=$ $r^{\rm (M)}_{\mathrm{R,D}}$ by properly designing the cooperative IRS passive beamforming (shown next).

Note that the optimal cooperative IRS passive beamforming for maximizing $r^{\rm (M)}_{\mathrm{S,R}}$ in \eqref{R_sr_C} is difficult to obtain due to the non-convex unit-modulus constraint as well as the coupling between $\boldsymbol{\Phi}_{\mathrm I_1,{\mathrm{S,R}}}$ and $\boldsymbol{\Phi}_{\mathrm I_2,{\mathrm{S,R}}}$ in the effective channel gain (see \eqref{hsrc}). Although a suboptimal solution can be obtained by designing an iterative algorithm using e.g., the semidefinite relaxation (SDR) and block coordinate descent (BCD) techniques \cite{9362274,9241706}, it yields little insight into the achievable rate under this IRS deployment strategy. To address this issue, we focus on characterizing the capacity scaling order for the multi-IRS deployment in the sequel.
We first derive the lower and upper bounds on the capacity given in \eqref{C_c}, and then obtain its scaling order when $M\to\infty$. 

\vspace{-5pt}
\subsection{Capacity Lower Bound}
First, we derive the lower bound on $C^{\mathrm{(M)}}$. 
For the effective channel in \eqref{hsrc}, it is known that when $M\to\infty$, the double-reflection channel dominates the others and the direct S-R channel is negligible \cite{9060923}.\footnote{Note that in \cite{9060923}, only the double-reflection link is considered.} 
Thus, to obtain the lower bound on $C^{\mathrm{(M)}}$, we propose to design the cooperative IRS passive beamforming to maximize the double-reflection channel power gain only, and at the same time, make the two single-reflection channels phase-aligned with the double-reflection channel. 
Note that it will be shown by simulations that this cooperative IRS passive beamforming design achieves favorable rate performance even when $M$ is small, and moreover does not compromise the capacity scaling order.
To this end, we first decompose the inter-IRS channel between I$_1$ and I$_2$, $\boldsymbol{G}_{\rm I_{1},\rm I_{2}}$, as follows.
\begin{equation}
    \begin{aligned}
        \boldsymbol{G}_{\rm I_{1},\rm I_{2}}\!\!=\!\!\underbrace{\sqrt{{g}_{\rm I_{1},\rm I_{2}}}\boldsymbol{a}_{\rm r}\!\!\left(\theta^{\rm r}_{\rm I_{1},\rm I_{2}},\!\vartheta^{\rm r}_{\rm I_{1},\rm I_{2}},\!M_2\right)}_{\boldsymbol{t}_{1}}\underbrace{\sqrt{{g}_{\rm I_{1},\rm I_{2}}}\boldsymbol{a}_{\rm t}^H\!\!\left(\theta^{\rm t}_{\rm I_{1},\rm I_{2}},\!\vartheta^{\rm t}_{\rm I_{1},\rm I_{2}},\!M_1\right)}_{\boldsymbol{t}^H_{2}}.
    \end{aligned}
\end{equation}
Then the optimization problem for designing the cooperative IRS passive beamforming can be formulated as
\begin{align}
    &\max_{\substack{\boldsymbol{\Phi}_{\mathrm I_1,{\mathrm{S,R}}}, \boldsymbol{\Phi}_{\mathrm I_2,{\mathrm{S,R}}}}}
    \quad~~|\boldsymbol{g}^H_{\mathrm{I_{2}, R}}\boldsymbol{\Phi}_{\mathrm{I_2,S,R}}\boldsymbol{G}_{\mathrm{I_{1}, I_{2}}}\boldsymbol{\Phi}_{\mathrm{I_1,S,R}}\boldsymbol{g}_{\mathrm{S, I_{1}}}|^2\\
    &\quad~~~~\text{s.t.} \qquad~~~~~~\boldsymbol{\Phi}_{\mathrm{I}_k,\mathrm{S},\mathrm{R}}\in\boldsymbol{\mathcal{P}}_{M_k},k=1,2,\nonumber\\
    &\quad~\qquad\qquad~~~~~~~\arg\!\left(h_{\mathrm{dr}}\right)\!=\!\arg\!\left(h_{\mathrm{sr,I}_k}\right),k=1,2.\nonumber
\end{align}
It can be shown that the optimal solution is given by
\begin{align}
    &[\boldsymbol{\hat\Phi}_{1}]_{m_1,m_1}\!\triangleq\!e^{\jmath\left({\arg\left({\boldsymbol{t}_1^H\boldsymbol{g}_{\mathrm{S,I_2}}}\right)}-\arg\left({[\boldsymbol{t}^H_{2}]_{m_1}}\right)-\arg\left({[\boldsymbol{g}_{\mathrm{S, I_1}}]_{m_1}}\right)\right)},\label{rc1}\\
    &[\boldsymbol{\hat\Phi}_{2}]_{m_2,m_2}\!\triangleq\!e^{\jmath\left(\arg\left({\boldsymbol{t}_2^H\boldsymbol{g}_{\mathrm{I_1,R}}}\right)-\arg\left({[\boldsymbol{g}^H_{\mathrm{I_2, R}}]_{m_2}}\right)-\arg\left({[\boldsymbol{t}_{1}]_{m_2}}\right)\right)},\label{rc2}
\end{align}
where $m_k=1,\cdots,M_k,k=1,2.$

Based on the above, we obtain the lower bound on $C^{\mathrm{(M)}}$.
\begin{lemma}\label{lem1}
\emph{For multi-IRS deployment, the system capacity given in \eqref{C_c} is lower-bounded by}
\begin{align}\label{C_c_low}
    C^{\mathrm{(M)}}\geq C^{\mathrm{(M)}}_{\rm low}\triangleq \frac{1}{2}\log_2\left(1+\frac{P}{\sigma^2} |h_{\mathrm{S,R,low}}^{\mathrm{(M)}}|^2\right),
\end{align}
\emph{where $h_{\mathrm{S,R,low}}^{\mathrm{(M)}}\triangleq\max\left(\frac{M_1 M_2 \beta^{\frac{3}{2}}_{0}}{(H_1H_2D_{\mathrm{I_1,I_2}})^{\frac{\alpha}{2}}}\!-\!\frac{\beta_0^{\frac{1}{2}}}{L^{\frac{\alpha}{2}}},0\right)$.}
\end{lemma}
\begin{proof}
By substituting the passive beamforming of I$_1$ and I$_2$ given in \eqref{rc1} and \eqref{rc2} into \eqref{hsrc}, we have
\hspace{-7pt}\begin{align}\label{ineq:hsr_low}
\allowdisplaybreaks[4]
\hspace{-7pt}&|h_{\mathrm{S,R}}^{\mathrm{(M)}}|\!\triangleq\!|g_{\mathrm{S,R}}\!\!\!+\!\boldsymbol{g}^H_{\mathrm{I_{2},R}}\!\boldsymbol{\hat\Phi}_{2}\boldsymbol{G}_{\mathrm{I_{1}\!,I_{2}}}\!\boldsymbol{\hat\Phi}_{1}\boldsymbol{g}_{\mathrm{S, I_{1}}}\!\!+\!\boldsymbol{g}^H_{\mathrm{I_{1}\!, R}}\!\boldsymbol{\hat\Phi}_{1}\boldsymbol{g}_{\mathrm{S, I_{1}}}\!\!\!+\!\boldsymbol{g}^H_{\mathrm{I_{2}\!,R}}\!\boldsymbol{\hat\Phi}_{2}\boldsymbol{g}_{\mathrm{S,I_{2}}}|\nonumber\\ \allowdisplaybreaks[4]
\hspace{-7pt}&\overset{(c)}{\geq} |\boldsymbol{g}^H_{\mathrm{I_{2},R}}\boldsymbol{\hat\Phi}_{2}\boldsymbol{G}_{\mathrm{I_{1}, I_{2}}}\!\boldsymbol{\hat\Phi}_{1}\boldsymbol{g}_{\mathrm{S,I_{1}}}\!\!+\!\boldsymbol{g}^H_{\mathrm{I_{1},R}}\boldsymbol{\hat\Phi}_{1}\boldsymbol{g}_{\mathrm{S,I_{1}}}\!\!+\!\boldsymbol{g}^H_{\mathrm{I_{2},R}}\boldsymbol{\hat\Phi}_{2}\boldsymbol{g}_{\mathrm{S,I_{2}}}|\!-\!|g_{\mathrm{S,R}}|\nonumber\\
\hspace{-7pt}&\overset{(d)}{=}\!|\boldsymbol{g}^H_{\mathrm{I_{2},R}}\!\boldsymbol{\hat\Phi}_{2}\boldsymbol{G}_{\mathrm{I_{1}\!,I_{2}}}\!\boldsymbol{\hat\Phi}_{1}\boldsymbol{g}_{\mathrm{S, I_{1}}}\!|\!+\!|\boldsymbol{g}^H_{\mathrm{I_{1},R}}\!\boldsymbol{\hat\Phi}_{1}\boldsymbol{g}_{\mathrm{S,I_{1}}}\!|\!+\!|\boldsymbol{g}^H_{\mathrm{I_{2},R}}\!\boldsymbol{\hat\Phi}_{2}\boldsymbol{g}_{\mathrm{S, I_{2}}}\!|\!-\!|g_{\mathrm{S,R}}|\nonumber\\
\hspace{-7pt}&\geq\!\! |\boldsymbol{g}^H_{\mathrm{I_{2},R}}\boldsymbol{\hat\Phi}_{2}\boldsymbol{G}_{\mathrm{I_{1}, I_{2}}}\!\boldsymbol{\hat\Phi}_{1}\boldsymbol{g}_{\mathrm{S,I_{1}}}|\!\!-\!\!|g_{\mathrm{S,R}}|\!=\!\!\frac{M_1 M_2 \beta^{\frac{3}{2}}_{0}}{(H_1\!H_2D_{\mathrm{I_1\!,I_2}})^{\frac{\alpha}{2}}}\!\!-\!\!\frac{\beta_0^{\frac{1}{2}}}{L^{\frac{\alpha}{2}}},
\end{align}
where $(c)$ is due to the triangle inequality, and $(d)$ is because the two single- and one double-reflection links are phase-aligned. 
Combining \eqref{C_c} and \eqref{ineq:hsr_low} yields the desired result in \eqref{C_c_low}.
\end{proof}

\vspace{-10pt}
\subsection{Capacity Upper Bound}
Next, we derive the upper bound on $C^{\mathrm{(M)}}$. To this end, we first define the \textit{favorable channel conditions} for the designed cooperative IRS passive beamforming given in \eqref{rc1} and \eqref{rc2}.
\vspace{3pt}

\noindent\underline{\textbf{\textit{Definition}}} \textbf{1.} (Favorable channel conditions) For multi-IRS deployment, with the cooperative IRS passive beamforming given in \eqref{rc1} and \eqref{rc2}, the channels, $h_{\mathrm{dr}}$, $h_{\mathrm{sr,I}_1}$ and $h_{\mathrm{sr,I}_2}$, are said to satisfy the favorable channel conditions when
\begin{align}
    &e^{\jmath\arg{\left([\diag{\left(\boldsymbol{g}^H_{\mathrm{S,I}_{k}}\right)}\boldsymbol{g}_{\mathrm{I}_{k},\mathrm{R}}]_{m_k}\right)}}=[\boldsymbol{\hat\Phi}_{k}]_{m_k,m_k},\label{fcc_1}\\
    &\arg(h_{\mathrm{dr}})=\arg(h_{\mathrm{sr,I}_1})=\arg(h_{\mathrm{sr,I}_2})=\arg(g_{\mathrm{S,R}}).\label{fcc_2}
\end{align}
When the favorable channel conditions in \eqref{fcc_1} and \eqref{fcc_2} hold, it can be shown that the proposed cooperative IRS passive beamforming in \eqref{rc1} and \eqref{rc2} simultaneously maximizes the channel power gains of both the two single-reflection and one double-reflection links, and at the same time, makes the above three reflection links phase-aligned with the direct link.
It is worth noting that the condition in \eqref{fcc_1} can be approximately satisfied when the IRS altitude is sufficiently small.
Next, we derive the upper bound on $C^{\mathrm{(M)}}$ under the favorable channel conditions.
\begin{lemma}\label{lem2}
\emph{For multi-IRS deployment, the system capacity given in \eqref{C_c} is upper-bounded by}
\begin{align}\label{C_c_up}
    C^{\mathrm{(M)}}\leq C^{\mathrm{(M)}}_{\rm upp}\triangleq \frac{1}{2}\log_2\left(1+\frac{P}{\sigma^2} |h^{\mathrm{(M)}}_{\mathrm{S,R,upp}}|^2\right),
\end{align}
\begin{equation}\nonumber
    \noindent\text{\emph{where}}~|h^{\mathrm{(M)}}_{\mathrm{S,R,upp}}|=\frac{\beta_0^{\frac{1}{2}}}{L^{\frac{\alpha}{2}}}+\frac{M_1 M_2 \beta^{\frac{3}{2}}_{0}}{(H_1\!H_2D_{\mathrm{I_1,I_2}})^{\frac{\alpha}{2}}}+\sum_{k=1}^2\frac{M_k\beta_{0}}{H_k^{\frac{\alpha}{2}}\!(L^2\!+\!H_k^2)^{\frac{\alpha}{4}}}.
\end{equation}
\end{lemma}
\begin{proof}
For the effective channel gain given in \eqref{hsrc}, we have
\vspace{-3pt}
\!\!\begin{align}\label{ineq:hsr_up}
\allowdisplaybreaks[4]
\!\!\!\!\!&\!\!\!\!\!\!|h^{\mathrm{(M)}}_{\mathrm{S,R}}|\!=\!|g_{\mathrm{S,R}}\!\!+\!h_{\mathrm{dr}}\!+\!h_{\mathrm{sr,I}_1}\!\!\!+\!h_{\mathrm{sr,I}_2}|\!\!\overset{(e)}{\leq}\!\!|g_{\mathrm{S,R}}|\!\!+\!\!|h_{\mathrm{dr}}|\!\!+\!\!|h_{\mathrm{sr,I}_1}\!|\!\!+\!\!|h_{\mathrm{sr,I}_2}|\!\nonumber\\
&\!\!\!\!\!\!\overset{(f)}{=}\!\!|h^{\mathrm{(M)}}_{\mathrm{S,R,upp}}|\!\!\triangleq\!\!\frac{\beta_0^{\frac{1}{2}}}{L^{\frac{\alpha}{2}}}\!\!+\!\!\frac{M_1 M_2 \beta^{\frac{3}{2}}_{0}}{(H_1\!H_2D_{\mathrm{I_1,I_2}})^{\frac{\alpha}{2}}}\!\!+\!\!\sum_{k=1}^2\frac{M_k\beta_{0}}{H_k^{\frac{\alpha}{2}}\!(L^2\!+\!H_k^2)^{\frac{\alpha}{4}}},
\end{align}
where $(e)$ is due to the triangle inequality, and $(f)$ is obtained by individually optimizing the passive beamforming for each of the single-reflection and double-reflection channels, which can be achieved by using the cooperative IRS passive beamforming given in \eqref{rc1} and \eqref{rc2} when the favorable channel conditions in \eqref{fcc_1} and \eqref{fcc_2} hold. 
Combining \eqref{C_c} and \eqref{ineq:hsr_up} yields the desired result.
\end{proof}

\vspace{-10pt}
\subsection{Capacity Scaling Order}
Last, based on Lemmas \ref{lem1} and \ref{lem2}, we obtain the capacity scaling order under the multi-IRS deployment strategy as follows.
\begin{proposition}\label{pro4}
\emph{For multi-IRS deployment, the system capacity increases with $M$ as $M\to\infty$ according to}
\vspace{-3pt}
\begin{align}\label{C_Ca}
    \lim_{M\to\infty}\frac{C^{\mathrm{(M)}}}{\log_2M}=2.
\end{align} 
\end{proposition}\vspace{-3pt}
\begin{proof}
With $M\to\infty$ and based on Lemma \ref{lem1}, we have
\vspace{-3pt}\vspace{-3pt}\!\!\begin{align}\label{Ccl}
    &\lim_{M\to\infty}C^{\mathrm{(M)}}\!\geq\!\lim_{M\to\infty}C^{\mathrm{(M)}}_{\rm low}=\lim_{M\to\infty}\!\frac{1}{2} \log _{2}\!\left(\!1\!+\!\frac{P}{\sigma^{2}}\!|h_{\mathrm{S,R,low}}^{\mathrm{(M)}}|^2\right)\nonumber\\
    &\!=\!\!\lim_{M\to\infty}\!\frac{1}{2} \log _{2}\!\left(\!1\!+\!\frac{P}{\sigma^{2}}\!\left(\!\frac{M_{1} M_{2} \beta_{0}^{\frac{3}{2}}}{(H_1\!H_2D_{\mathrm{I_1,I_2}})^{\frac{\alpha}{2}}}\!-\!\frac{\beta_{0}^{\frac{1}{2}}}{L^{\frac{\alpha}{2}}}\!\right)^{\!\!2}\right)\\
    &\!=\!\!\lim_{M\to\infty}\!\!\log_2\!\!\left(\!\!\frac{\sqrt{P}\beta_0^{\frac{3}{2}}M_1M_2}{\sigma{(H_1\!H_2D_{\mathrm{I_1,I_2}})^{\frac{\alpha}{2}}}}\!\!\right)\!\overset{(g)}{=}\!\!\!\!\lim_{M\to\infty}\!\!\log_2\!\!\left(\!\!\frac{\sqrt{P}\beta_0^{\frac{3}{2}}\rho(1\!-\!2\rho)M^2}{\sigma{(H_1\!H_2D_{\mathrm{I_1,I_2}})^{\frac{\alpha}{2}}}}\!\!\right)\nonumber,
\end{align}
where $(g)$ is obtained by substituting $M_1=\rho M$ and $M_2=(1-2\rho) M$.
Similarly, based on Lemma \ref{lem2}, we have
\vspace{-3pt}
\begin{align}\label{Ccu}
    &\lim_{M\to\infty}\!C^{\mathrm{(M)}}\!\leq\!\lim_{M\to\infty}\!C^{\mathrm{(M)}}_{\rm upp}\!=\!\!\lim_{M\to\infty}\!\frac{1}{2} \log _{2}\!\left(\!1\!+\!\frac{P}{\sigma^{2}}\!|h_{\mathrm{S,R,upp}}^{\mathrm{(M)}}|^2\right)\nonumber\\
    &=\lim_{M\to\infty}\log_2\!\left(\frac{\sqrt{P}\beta_0^{\frac{3}{2}}\rho(1-2\rho)M^2}{\sigma{(H_1\!H_2D_{\mathrm{I_1,I_2}})^{\frac{\alpha}{2}}}}\right).
\end{align}\vspace{-3pt}
Combining \eqref{Ccl} and \eqref{Ccu} leads to
\begin{equation}
\lim_{M\to\infty}\!C^{\mathrm{(M)}}_{\rm low}\leq \lim_{M\to\infty}C^{\mathrm{(M)}}\leq \lim_{M\to\infty}\!C^{\mathrm{(M)}}_{\rm upp},
\end{equation}
where $\lim_{M\to\infty}\!C^{\mathrm{(M)}}_{\rm low}=\lim_{M\to\infty}\!C^{\mathrm{(M)}}_{\rm upp}$.
Thus, we have\vspace{-3pt}
\begin{align}\label{C_C_asy}
    &\lim_{M\to\infty}\frac{C^{\mathrm{(M)}}}{\log_2M}\!=\!\lim_{M\to\infty}\frac{2\log_2M}{\log_2M}\!+\!\frac{\log_2\left(\sqrt{P}\beta_0^{\frac{3}{2}}\rho(1\!-\!2\rho)\right)}{\log_2M}\nonumber\\
    &\!-\!\frac{\log_2\left(\sigma{(H_1\!H_2D_{\mathrm{I_1,I_2}})^{\frac{\alpha}{2}}}\right)}{\log_2M}
    =2,
\end{align}
which thus completes the proof.
\end{proof}

Proposition \ref{pro4} is expected since the capacity with an asymptotically large $M$ is dominated by the double-reflection link, which has a capacity scaling order of $2$.
Moreover, it is worth noting that given a fixed number of reflecting elements, $M$, the asymptotic capacity with multi-IRS deployment in \eqref{C_C_asy} is maximized when $\rho=\frac{1}{4}$ (i.e., $M_1=M_3=\frac{M}{4},M_2=\frac{M}{2}$), which can be proved by setting the first-order derivative of the term $\log_2\left(\frac{\sqrt{P}\beta_0^{\frac{3}{2}}\rho(1-2\rho)}{\sigma{(H_1\!H_2D_{\mathrm{I_1,I_2}})^{\frac{\alpha}{2}}}}\right)$ given in \eqref{C_C_asy} over $\rho$ equal to 0.
\end{spacing}

\vspace{-3pt}
\begin{spacing}{0.85}
\section{Numerical Results}
Numerical results are presented in this section. {\color{black}The distance from S to D is set to be 1000 m with the R deployed at the middle between S and D, i.e., $L=500$ m, and the IRS near S (D) is deployed with the downtilt angle of $\theta=\frac{\pi}{4}$. Moreover, for the IRSs near R and S (D), we set their altitudes as $H_1=5$ m and $H_2=4$ m, respectively.} The system operates at a carrier frequency of 6 GHz with the wavelength of $\lambda=0.05 \mathrm{~m}$. Other parameters are set as $\beta_0=-30$ dB, $\alpha=2$, $P=30$ dBm, $d_{\rm I}=\frac{\lambda}{4}$, and $\sigma^2=-90$ dBm.

\begin{figure}[t] \centering
{\subfigure[\color{black}{LoS channel model.}] {
\label{capa}
\includegraphics[width=0.46\linewidth]{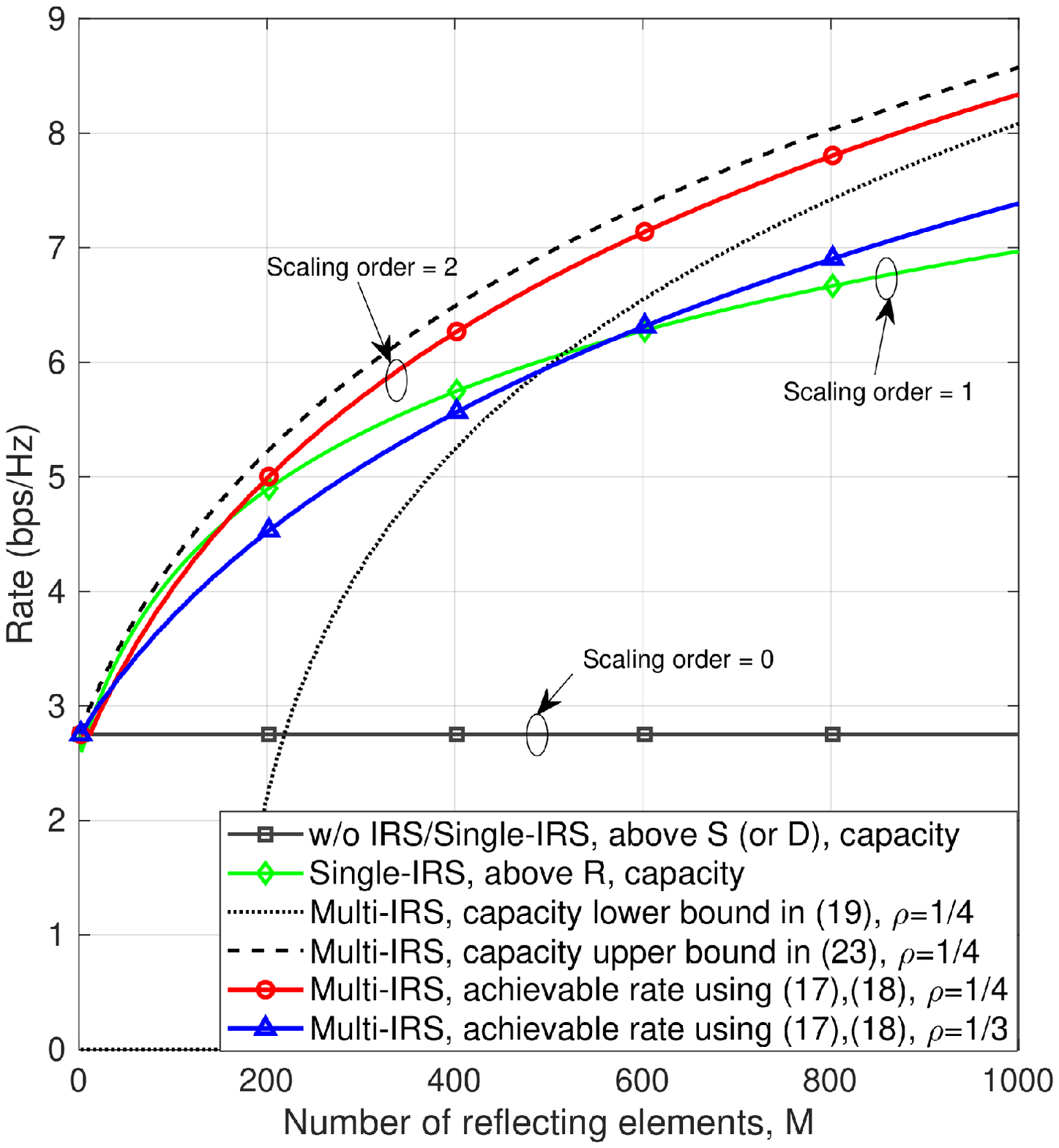}  
}}
{\subfigure[\color{black}{Rician fading channel model.}] {\label{rician}
\includegraphics[width=0.46\linewidth]{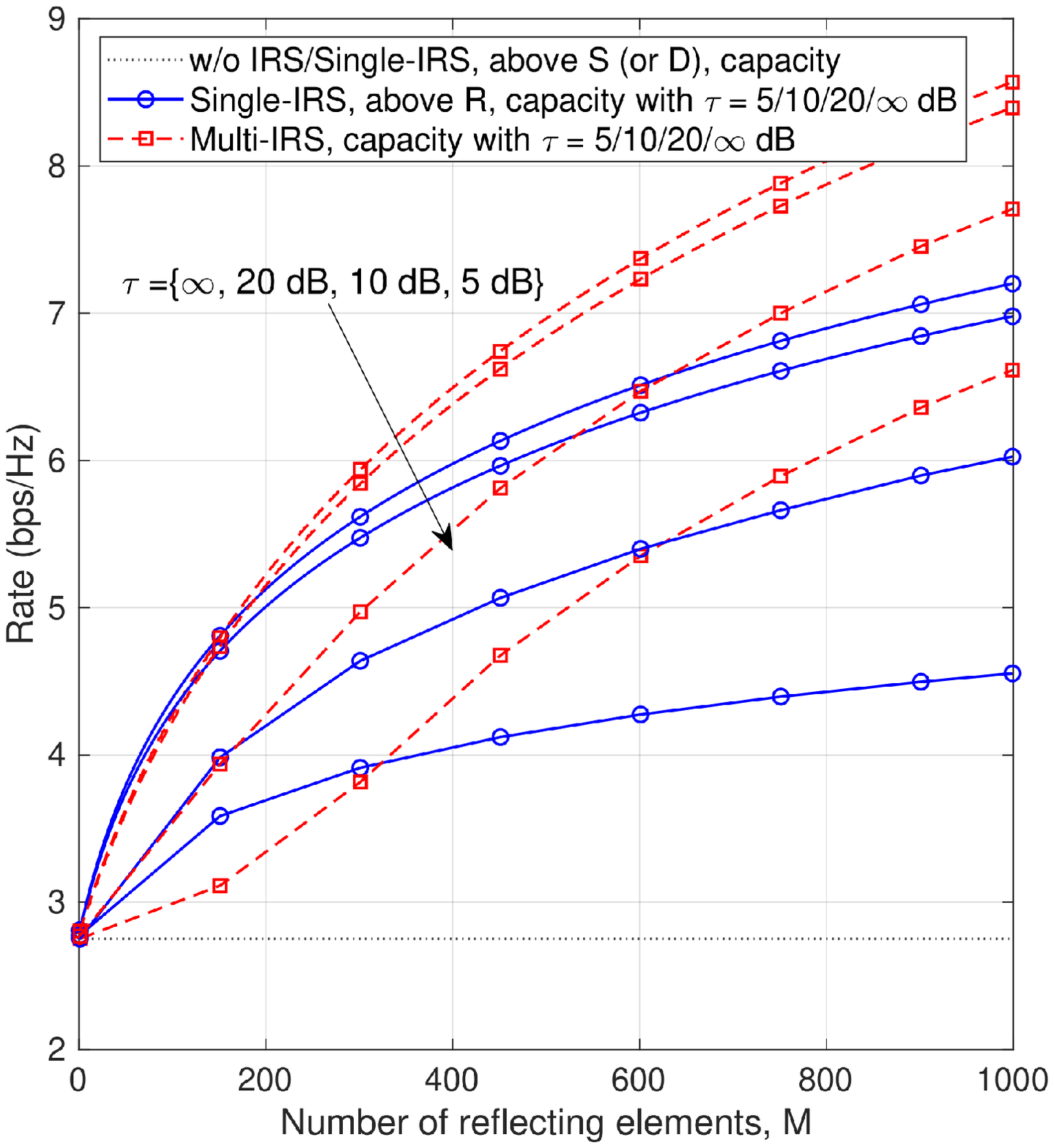}
}}
{\caption{{\color{black}Rate comparison between single- and multi-IRS deployment strategies.}}}
\end{figure}

Fig.~\ref{capa} compares the rates under the single- and multi-IRS deployment strategies. 
{\color{black}First, for single-IRS deployment, it is observed that deploying the IRS near R achieves a higher capacity than that near S (or D) or that without IRS (w/o IRS), which is consistent with Corollary \ref{pro3}.} 
Second, for multi-IRS deployment, it is observed that the capacity lower bound approaches its upper bound as $M$ increases.
Moreover, our proposed cooperative IRS passive beamforming design given in \eqref{rc1} and \eqref{rc2} is shown to achieve close rate performance to its capacity upper bound, especially when $M$ is large. 
Another observation is that the multi-IRS deployment strategy with $\rho=\frac{1}{4}$ outperforms that with equal elements allocation among the three IRSs, i.e., $\rho=\frac{1}{3}$.
{\color{black}Last, we observe that the proposed multi-IRS deployment strategy achieves a higher capacity than the single-IRS deployment near R or S/D when $M$ is sufficiently large (e.g., $M=150$), thanks to the higher asymptotic passive beamforming gain of the double-reflection links.}

{\color{black}In Fig. \ref{rician}, we evaluate the rate performance of different IRS deployment strategies under a more practical scenario, where the channel from node $i$ to $j$, with $i\in\{\rm{S,I~ (I_1),R}\}$, $j\in\{\rm{I~(I_2),R,D}\}$, $i\neq j$, and $\{i,j\}\neq\{\rm{S,D}\}$, follows the Rician fading channel model with its Rician factor denoted by $\tau$. Note that we optimize the IRS deployment strategies and their corresponding beamforming designs based on the (dominant) LoS component by ignoring the non-LoS (NLoS) components at first. It is observed that as the Rician factor $\tau$ increases, the achievable rates of both the single- and multi-IRS deployment strategies increase. Moreover, the case with a sufficiently large $\tau$ (e.g., $\tau=20$ dB) achieves close performance to that under the (ideal) LoS channel model (i.e., $\tau=\infty$). Other observations are similar to those in Fig. \ref{capa}.}

\end{spacing}
\begin{spacing}{0.9}
\vspace{-2pt}
\section{Conclusions}

In this letter, we studied the IRS deployment strategy for an IRS-aided wireless relaying system. 
Under the LoS channel model, we characterize the capacity scaling orders of the single- and multi-IRS deployment strategies.
It was shown that the multi-IRS deployment strategy achieves a higher capacity than the single-IRS counterparts when the total number of reflecting elements is large thanks to the higher passive beamforming gain of double-reflection links as compared to single-reflection links.
Moreover, numerical examples validated our analytical results and the effectiveness of our proposed cooperative IRS passive beamforming design.
\end{spacing}
\vspace{-2pt}
\bibliographystyle{IEEEtran}

\end{document}